\documentclass[a4paper,11pt]{article}
\usepackage{latexsym}	
\usepackage{amsmath}
\usepackage{amsthm}

\usepackage{ amssymb }
\def\ba{\begin{eqnarray}}
\def\ea{\end{eqnarray}}

\def\R{\hbox{\bf R}}

\theoremstyle{plain}
\newtheorem{theorem}{Theorem}
\newtheorem{prop}[theorem]{Proposition}
\newtheorem{lemma}{Lemma}[subsection]

\def\R{\mathbb R}
\begin{document}
\baselineskip0.25in
\title{The existence of smooth solutions in q-models}
 \author{Juliana Osorio Morales \thanks{Departamento de Matem\'aticas Luis Santal\'o (IMAS), UBA CONICET, Buenos Aires, Argentina
juli.osorio@gmail.com.} and Osvaldo P. Santill\'an \thanks{Departamento de Matem\'aticas Luis Santal\'o (IMAS), UBA CONICET, Buenos Aires, Argentina
firenzecita@hotmail.com and osantil@dm.uba.ar.}}

\date {}
\maketitle

\begin{abstract}

The q-models are scenarios that may explain the smallness of the cosmological constant \cite{volo1}-\cite{volo5}. The vacuum in these theories is presented as a self-sustainable medium and include a new degree of freedom, the q-variable, which stablish the equilibrium of the quantum vacuum. In the present work, the Cauchy formulation for these models is studied. It has been already noted that there exist some limits where these theories are described by an $F(R)$
model, which posses a well formulated Cauchy problem. This paper shows that the Cauchy problem is well posed even not reaching this limit. By use of some mathematical theorems about second order non linear systems, it is shown that these scenarios admit a smooth solution for at least a finite time when some specific type of initial conditions are imposed. Some technical conditions of \cite{ringstrom} play an important role in this discussion. 

\end{abstract}

\section{Introduction}

One of the current problems in cosmology is the explanation of the cosmic acceleration of the visible universe \cite{acelero1}-\cite{acelero3}.  The problem is that gravity is an attractive force and, therefore, deceleration will be expected instead of acceleration. 
Another unsolved problem is the discrepancy between the luminous matter of several objects in the universe and
their gravitational effects \cite{dm1}-\cite{dm2}. As there is evidence that the universe is almost flat, the current energy density should be of the order of the critical one, $\rho_c\sim 10^{-47}$GeV$^4$
\cite{dm3}. But this does not agree with the contributions corresponding to the
dynamically measured non relativistic mass density, which is approximately $(0.1-0.3)\rho_c$. 

The acceleration of the
universe expansion suggests the presence of a cosmological constant  \cite{acelero1}-\cite{acelero3}. If this were to be interpreted
as vacuum energy density, then its value would be a considerable fraction of the critical density $\rho_c$.
In addition, one  of the plausible scenarios for explaining the discrepancy between the luminous and gravitational mass density is the existence of dark matter, which is an unknown matter
sector whose contribution to the energy density compensates the difference \cite{dm1}-\cite{dm2}.

 One of theoretical problems of the expected value of the cosmological constant is that the QFT predictions of vacuum energy
 are at least 55 orders of magnitude larger than $\rho_c$ \cite{dolgo2}.  Thus, there exist two problems to be explained namely, why the energy density of the universe is so small and why it is so close to the critical one $\rho_c$. 

There exists an approach for solving the first problem, which assume the existence of an unknown matter component whose evolution screens the QFT energy density at late times. Examples of these scenarios are given in \cite{dolgin}- \cite{urbano1}.
Another type of models are the self-tuning vacuum variable scenarios \cite{volo1}. The idea behind these models is that the vacuum is a self-sustainable medium, that is, it has a definite volume even in an empty environment. These works postulate a new degree of freedom, called the q-variable, whose role is  the equilibration of the quantum vacuum. Other thermodynamical scenarios of this type were 
 considered in \cite{volo2}-\cite{volo5}. The q-scenarios inspired partially the interest in vector  vector fields adjustment mechanisms, or even tensor ones.  Some vector models capable to adjust vacuum energy to a very low value 
were presented  several years ago in \cite{dolgo1}-\cite{dolgo2}. These models spontaneously break the Lorentz symmetry, and can be considered as particular cases of more general models considered by Bjorken \cite{bjorken1}-\cite{bjorken2}. A first obstacle in the original formulation  \cite{dolgo1}-\cite{dolgo2} is that  the effective Newton constant $G_N$
obtains an unacceptable numerical value \cite{rubak}. Furthermore, they strongly modified the dispersion relation of gravitational waves and introduce longitudinal components wide beyond the experimental accuracy \cite{rubak}. However, there exist new scenarios which apparently overcome this problem \cite{emelo0}-\cite{emelo4}. It was suggested recently that these models admit an plausible Newtonian limit and give rise to reasonable gravitational waves in \cite{scorna}.

The present work is inspired in the q-models described above, but is not focused in the cosmological constant side. Instead, attention is paid to their Cauchy formulation. The aim is to show that, when some suitable initial conditions are formulated on a Cauchy surface  and that the evolution is globally hyperbolic, the solution exists for a proper time $\tau>0$ and is smooth (continuous with all its derivatives continuous). It will be shown that  Cauchy problem is well formulated on the Einstein frame.

The structure of the present work is as follows. In section 2 the main equations of the model are derived. These equations are presented in the Einstein frame in section 3. In section 4 it is shown that the resulting system takes the form considered in \cite{foures}, and it can be shown that there exist a $C^5$ solution when the initial conditions are of $C^4$ type. The proof of the existence of an smooth solution is given in section 5 and the appendix. In section 5, some theorems due to Ringstrom \cite{ringstrom} are employed to show that the resulting system admits solutions with an specific condition, namely, x-compactness. This condition is not satisfyied for a Lorenzian metric in a globally hyperbolic space time. For this reason, this solution should be interpreted as a local one. However, in the appendix it is shown that it is possible to glue the local solutions of x-support found to a global one. It is important to remark that these techniques were used for GR with an scalar field by Ringstrom \cite{ringstrom}, however the scalar field Ringstrom considers has canonical kinetic energy. Instead, in our formulation in the Einstein frame the resulting theory possess two scalar fields with non canonical kinetic terms. For this reason we analysed carefully the mathematical structure that leads to the Ringstrom results and were able to apply them even when the kinetic terms acquire the non canonical form. Finally, in section 6 the obtained results and further perspectives are discussed.

\section{The equations defining the model}

The present section follows closely the original references \cite{volo1}-\cite{volo5}.
The model to be considered contains several conserved microscopic variables
$q^{(a)}$, for $a=1$, $\dots$ , $n$, with their corresponding chemical
potentials $\mu^{(a)}$. The variables
 $q^{(a)}$ can be represented by a four-form field $F^{(a)}_{\mu\nu\rho\sigma}$.
Define the scalar quantity
$$
(\phi_a)^2 = -  \frac{1}{24}\,
F^{(a)}_{\mu\nu\rho\sigma}\, F^{(a)\mu\nu\rho\sigma},\qquad
F^{(a)}_{\mu\nu\rho\sigma}=
\nabla^{\phantom{(a)}}_{[\mu}\!\!A^{(a)}_{\nu\rho\sigma]}\,.
$$
The matter field of the theory is assumed to be an scalar field $\psi(x)$.
The action of the four-form fields $F^{(a)}(x)$, the matter field $\psi(x)$,
and the gravitational field $g_{\mu\nu}(x)$ is given by
\begin{equation}
\label{qaction}
S[A^{(a)}, g,\psi]=
-\int_{\mathbb{R}^4} \,d^4x\, \sqrt{-g}
\, \left[
K(\phi_a) R   + \epsilon(\phi_a,\psi)
+ \frac{1}{2}\,\partial_\mu \psi\,\partial^\mu \psi
\right]\,.
\end{equation}
Here $ \epsilon(\phi_a,\psi)$ at the moment is an unspecified interaction. In addition 
$K(\phi_a)$ represents a coupling between the fields $F^a$ and the curvature $R$ of the space time.
This coupling does not depend on the matter field $\psi(x)$.
Starting with this action, the following equations for $F^a$ and $\psi$ are obtained
$$
\nabla_\mu \left[\frac{\sqrt{-g}\,F^{(a)\mu\nu\rho\sigma}}{\phi_a}
\left(  \frac{\partial\epsilon}{\partial \phi_a}
       +  R \,\frac{\partial K}{\partial \phi_a}\right)
                 \right]=0,
$$                 

\begin{equation}\label{fieldeqs}
\square\psi -  \frac{\partial\epsilon}{\partial\psi} =0.
\end{equation}
Here $\square=g^{\mu\nu}\nabla_\mu\nabla_\nu$ is the standard laplacian in four dimensions. 
The equations describing the gravitational field $g_{\mu\nu}$ are given by
\begin{equation}\label{einstein1}
2K
\left( R_{\mu\nu}-\frac{1}{2}\,R\,g_{\mu\nu}\right)
+R\, g_{\mu\nu} \sum_{a=1}^n\, \phi_a\,\frac{\partial K}{\partial \phi_a}
+2
\Big(  \nabla_\mu\nabla_\nu - g_{\mu\nu}\,\square \Big)K
-\widetilde{\epsilon}(\phi_a,\psi)\, g_{\mu\nu}
+T^\text{M}_{\mu\nu} =0,
\end{equation}
where the following effective interaction
$$
\widetilde{\epsilon}(\phi_a,\psi)\equiv \epsilon(\phi_a,\psi)
-\sum_{b=1}^{n} \phi_{(b)}\,\frac{\partial\epsilon}{\partial \phi_{(b)}}  
$$
and the following scalar-field energy-momentum tensor
$$
T^\text{M}_{\mu\nu} =
\partial_\mu \psi\,\partial_\nu \psi
- \frac{1}{2}\,g_{\mu\nu}\,\partial_\rho \psi\,\partial^\rho \psi,
$$ 
were introduced. These equations are of the Einstein type when the coupling $K(\phi_a)$ is replaced by a simple constant, otherwise they are not. But it will be shown in the next section that there exists a conformal transformation which renders the system to one of the Einstein type, with two scalar fiels with non canonical kinetic terms.

From the definition of the quantities $F^a$ it follows that the Maxwell equations may be expressed as follows
\begin{equation}\label{maxwell}
\partial_\mu
\left(   \frac{\partial\epsilon}{\partial \phi_a}
       + R\,\frac{\partial K}{\partial \phi_a}\right)=0\,.
\end{equation}
These are $4n$ equations with solutions
$$
     \frac{\partial\epsilon}{\partial \phi_a}
   + R \,\frac{\partial K}{\partial \phi_a}
 =\mu^{(a)},
$$
where the $\mu^{(a)}$ are $n$ integration constants, interpreted as chemical potentials in the original literature \cite{volo1}-\cite{volo5}.
After eliminating $\partial K/\partial \phi_a$,
one finds from the generalized Einstein equations that
\begin{equation}\label{einstein2}
-2K\Big( R_{\mu\nu}-\frac{R}{2}g_{\mu\nu} \Big)
-2\Big(  \nabla_\mu\nabla_\nu - g_{\mu\nu}\, \square\Big)\, K
+\Big(\epsilon-\sum_{a=1}^n \mu^{(a)} \phi_a \Big)\, g_{\mu\nu}
=T^\text{M}_{\mu\nu}\,.
\end{equation}
These equations may be obtained from the following effective action
\begin{equation}\label{qaction2}
S_\text{eff}[A^{(a)},\mu^{(a)},g,\psi]=
-\int_{\mathbb{R}^4} \,d^4x\, \sqrt{-g}\, \left(
 K\, R  + \epsilon-\sum_{a=1}^n \mu^{(a)} \phi_a
+ \frac{1}{2}\,\partial_\mu \psi\,\partial^\mu \psi\right)\,.
\end{equation}
The  $\mu^{(a)} \phi_a$ terms in this action do  not contribute to the equations of motion.
This follows from the fact that
$$
\int_{\mathbb{R}^4} \,d^4x\; \sqrt{|g|}\, \mu^{(a)}\, \phi_a =
-\frac{\mu^{(a)}}{24} ~e^{\kappa\lambda\mu\nu}
\int_{\mathbb{R}^4} \,d^4x\; F^{(a)}_{\kappa\lambda\mu\nu} \,.
$$
The constant $\mu^{(a)}$ is seen to play the role
of a Lagrange multiplier
related to the conservation of the vacuum charge $\phi_a$ \cite{volo1}-\cite{volo5}.

The potentially large microscopic energy density $\epsilon(\phi_a,\psi)$
in the original action
has been replaced by the vacuum energy density
enters the effective action namely,
$$
 \rho_\text{V}\equiv \epsilon(\phi_a,\psi)-\sum_{a=1}^n\mu^{(a)} \phi_a \,.
$$ 
This density may be considerably smaller than the bare vacuum energy, and for this reasons these models are of interest
in the context of cancellation or adjustment of the cosmological constant.

  It was shown in the original literature that, in certain limit, this model is effectively described by an $F(R)$
  theory. It was showed in the literature that $F(R)$ models posses a well posed Cauchy problem \cite{vignolo1}-\cite{vignolo2}.
  The aim of the present work is to show that q-models posses a well posed Cauchy problem even not when the $F(R)$
  limit is reached. Several mathematical results about quasi linear hyperbolic systems will be of importance in establishing these results.

\section{The model in the Einstein frame}

Experience with the $F(R)$ models of gravity suggest that the best way to analyse the Cauchy problem for these theories is to go to the Einstein frame. This is obtained by following  conformal transformation $\Omega^{2}=\frac{2K}{M_{p}^{2}}$, with $K=K(\phi_{1},\dots\phi_{n})$ the coupling of the curvature $R$ to the fields. After this conformal transformation the Lagrangian corresponding to the action (\ref{qaction2}) takes the form
$$
\mathcal L=\frac{M_{p}^{2}}{2}R-\frac{3}{2}g^{\mu\nu}\nabla_{\mu}\ln(\frac{2K}{M_{p}^{2}})\nabla_{\nu}\ln(\frac{2K}{M_{p}^{2}})\frac{M_{p}^{2}}{2}-\frac{M_{p}^{2}}{4K}\partial_{\mu}\psi\partial ^{\mu}\psi-\frac{M_{p}^{4}}{4K^{2}}(\epsilon-\mu^{(a)}\phi_{a}),
$$
with $ a=1,\dots,n$. By making the further field redefinition
$$
\eta=\frac{\sqrt{3}M_{p}}{2}\ln(\frac{2K}{M_{p}^{2}}),\qquad \longleftrightarrow \qquad K=\frac{M_{p}^{2}}{2}\exp(\frac{2\eta}{\sqrt{3}M_{p}}),
$$
the last lagrangian may be expressed as follows
\begin{equation}\label{qaction3}
\mathcal L=\frac{M_{p}^{2}}{2}R-\frac{1}{2}\partial_{\mu}\eta\partial^{\mu}\eta-\frac{1}{2}\exp(-\frac{2\eta}{\sqrt{3}M_{p}})\partial_{\mu}\psi\partial^{\mu}\psi-\exp(\frac{4\eta}{\sqrt{3}M_{p}})(\epsilon-\mu^{(a)}\phi_{a}).
\end{equation}
This model corresponds to two scalar fields $\eta$ and $\psi$, one of them with a non canonical kinetic term. The kinetic terms of both scalar fields can be combined into a two dimensional sigma model target space with metric $g_{ij}=$diag$(1, \exp(-\frac{2\eta}{\sqrt{3}M_{p}}))$ with $i=1,2$. The function 
\begin{equation}\label{efectivo}
U_{eff}=\exp(\frac{4\eta}{\sqrt{3}M_{p}})(\epsilon-\mu^{(a)}\phi_{a}),
\end{equation}
represents the effective potential for the scalar fields.

There is a one to one correspondence between the fields $K$ and $\eta$. Both fields are functions of all the $n$ auxiliary fields $\phi_{a}$. On the other hand one may consider one of the fields, say $\phi_{n}$, as a function of $\phi_{1},\dots,\phi_{n-1},\eta$. This may be possible only locally, by assuming that the coupling $K$ is a  function of the fields $\phi_a$ for which the implicit function theorem applies. A favourable possibility may be when one of the fields, say $\phi_n$ is related to the others and to $K$ by a relation with a single branch. If this situation is not realized for these fields, care should be taken concerning the following discussion.

 If the field redefinition of the previous paragraph is possible, then the constraints of the model 
$$
\frac{\delta\mathcal L}{\delta \phi_{a}}=0, \ \  \ a=1,\dots,n-1
$$
are equivalent to the following system
\begin{equation}\label{vinculo}
\frac{\partial \epsilon}{\partial\phi_{a}}+\frac{\partial\epsilon}{\partial\phi_{n}}\frac{\partial\phi_{n}}{\partial\phi_{a}}=\mu^{a}.
\end{equation}
Note that now there are $n-1$ constraints, since $\phi_n$ is now locally a function of $\eta$ and this last field is dynamical in the new reference frame. 

Concerning the dynamical fields, the equation of motion $\frac{\delta\mathcal L}{\delta \eta}=\nabla_{\mu}(\frac{\delta \mathcal L}{\delta\partial_{\mu}\eta})$ is equivalent to
\begin{equation}\label{eq for eta}
-\square\eta=\frac{1}{\sqrt{3}M_{p}}\exp{(-\frac{2\eta}{\sqrt{3}M_{p}})}\partial_{\mu}\psi\partial^{\mu}\psi+\frac{4}{\sqrt{3}M_{p}}\exp(-\frac{4\eta}{\sqrt{3}M_{p}})(\epsilon-\mu^{(a)}\phi_{a})
\end{equation}
\begin{equation}\label{eq for eta2}
-\exp(-\frac{4\eta}{\sqrt{3}M_{p}})(\frac{\partial\epsilon}{\partial\eta}-\mu^{(n)}\frac{\partial\phi_{n}}{\partial\eta}).
\end{equation}

On the other hand, the equations $\frac{\delta\mathcal L}{\delta \psi}=\nabla_{\mu}(\frac{\delta \mathcal L}{\delta\partial_{\mu}\psi})$ are
\begin{equation}\label{eq for psi}
\square\psi-\frac{2}{\sqrt{3}M_{p}}\partial_{\mu}\eta\partial^{\mu}\psi=\exp(-\frac{2\eta}{\sqrt{3}M_{p}})\frac{\partial\epsilon}{\partial\psi}.
\end{equation}

Finally, the variation with respect to the metric gives us the Einstein's equations
$$
R_{\mu\nu}-\frac{1}{2}g_{\mu\nu}R=G_{N}\left[\partial_{\mu}\eta\partial_{\nu}\eta+\exp(-\frac{2\eta}{\sqrt{3}M_{p}})\partial_{\mu}\psi\partial_{\nu}\psi-\frac{1}{2}g_{\mu\nu}g^{\alpha\beta}\partial_{\alpha}\eta\partial_{\beta}\eta\right.
$$
\begin{equation}\label{einstein}
\left.-\frac{1}{2}g_{\mu\nu}g^{\alpha\beta}\exp(-\frac{2\eta}{\sqrt{3}M_{p}})\partial_{\alpha}\psi\partial_{\beta}\psi-\exp(-\frac{4\eta}{\sqrt{3}M_{p}})g_{\mu\nu}(\epsilon-\mu^{a}\phi_{a})\right]
\end{equation}
The energy-momentum tensor of the matter fields can be read from the last expression, the result is
$$
T_{\mu\nu}=2\left(\frac{1}{2}\partial_{\mu}\eta\partial_{\nu}\eta+\frac{1}{2}\exp(-\frac{2\eta}{\sqrt{3}M_{p}})\partial_{\mu}\psi\partial_{\nu}\psi\right)
$$
$$
-g_{\mu\nu}\left(\frac{1}{2}\partial_{\alpha}\eta\partial^{\alpha}\eta+\frac{1}{2}\exp(-\frac{2\eta}{\sqrt{3}M_{p}})\partial_{\alpha}\psi\partial^{\alpha}\psi+\exp(-\frac{4\eta}{\sqrt{3}M_{p}})(\epsilon-\mu^{a}\phi_{a})\right).
$$
Its trace is 
$$
T=-G_{N}\left[g^{pq}\partial_{p}\eta\partial_{q}\eta+\exp(-\frac{2\eta}{\sqrt{3}M_{p}})g^{pq}\partial_{p}\psi\partial_{q}\psi+4\exp(-\frac{4\eta}{\sqrt{3}M_{p}})(\epsilon-\mu^{a}\phi_{a})\right]
$$
In these terms the Einstein equation may be expressed equivalently as $R_{\mu\nu}=T_{\mu\nu}-\frac{1}{2}g_{\mu\nu}T$. Explicitly
\begin{equation}\label{einstein2}
R_{\mu\nu}=G_{N}\left[\partial_{\mu}\eta\partial_{\nu}\eta+\exp(-\frac{2\eta}{\sqrt{3}M_{p}})\partial_{\mu}\psi\partial_{\nu}\psi+\exp(-\frac{4\eta}{\sqrt{3}M_{p}})g_{\mu\nu}(\epsilon-\mu^{a}\phi_{a})\right].
\end{equation}
The main task is to present some existence theorems related to the derived system of equations (\ref{vinculo})-(\ref{einstein2}).

\section{The equations of the model as a quasi linear system}
\subsection{The existence of $C^4$ solutions}
In this section the previous system will be expressed in the form of a quasi linear hyperbolic system. The advantage of doing this is that several results about these type of systems are known in the literature. A quasi linear hyperbolic system is one of the form
\begin{equation}\label{choquet}
A^{\mu\nu}(x, t, u_i)\frac{\partial u_q}{\partial x^\mu \partial x^\nu}=f_q(u_i,\partial u_i),
\end{equation}
where $u_q$ with $q=1,..,n$ constitute the $n$-unknowns. Here the matrix $A^{pq}$ is the same for all the equations $q=1,..n$ and it is of normal hyperbolic type, that is
$A_{44}\leq  0$ and $A_{ij}x^i x^j$ is a positive definite form, with the latin indices indicating spatial directions. 

There are several steps required in order to convert the system of equations of the previous section into one of the form (\ref{choquet}). For this, it is mandatory to choose a coordinate system. One choice that it is convenient for this purpose is the harmonic gauge $\square x^\mu=0$, which implies that the quantity $F^{i}=g^{pq}\Gamma_{pq}^{i}=0$.  In this gauge, the Ricci tensor for a generic metric $g_{\mu\nu}$ can be written as
\begin{equation}\label{desco}
R^{\mu\nu}=-\frac{1}{2}g^{\alpha\beta}\partial_{\alpha}\partial_{\beta}g^{\mu\nu}+Q^{\mu\nu}(g,\partial g),
\end{equation}
where the quantity
$$
Q^{\mu\nu}=g^{\alpha\beta}[\Gamma_{\alpha\gamma}^{\mu}\partial_{\beta} g^{\nu\gamma}+\Gamma_{\alpha\gamma}^{\nu}\partial_{\beta} g^{\mu\gamma}-2\Gamma_{\alpha\beta}^{\gamma}\partial_{\gamma} g^{\nu\mu}],
$$
has been introduced. That means that the Ricci tensor in these coordinates becomes a quasi-diagonal second-order operator for the components of $g$. In these terms the Einstein equation (\ref{einstein2}) may be expressed as
\begin{equation}\label{be}
g^{\alpha\beta}\partial_{\alpha}\partial_{\beta}g_{\mu\nu}=B_{\mu\nu}(g,\eta,\psi,\phi_{a},\partial\eta,\partial \psi),
\end{equation}
with $B_{\mu\nu}$ given explicitly by
$$
B_{\mu\nu}=G_{N}\left[\partial_{\mu}\eta\partial_{\nu}\eta+\exp(-\frac{2\eta}{\sqrt{3}M_{p}})\partial_{\mu}\psi\partial_{\nu}\psi+\exp(-\frac{4\eta}{\sqrt{3}M_{p}})g_{\mu\nu}(\epsilon-\mu^{a}\phi_{a})\right]-Q_{\mu\nu}(g,\partial g).
$$
On the other hand, by taking into account that 
$$
\square\psi=g^{\mu\nu}\partial_\mu \partial_\nu \psi+\Gamma^\alpha \partial_\alpha \psi,\qquad \Gamma^\alpha=\frac{1}{\sqrt{|g|}}\frac{\partial}{\partial x^\beta}(\sqrt{|g|}g^{\alpha\beta}),
$$
and the analogous formulas for $\eta$, it follows that the equations (\ref{eq for eta}), (\ref{eq for psi}), (\ref{vinculo}) reduce to the following system of second order partial differential equations
\begin{align}\label{hyper system}
g^{\alpha\beta}\partial_{\alpha}\partial_{\beta}g_{\mu\nu}=&B_{\mu\nu}(g,\eta,\psi,\phi_{a},\partial g, \partial\eta,\partial \psi),\notag\\
g^{\alpha\beta}\partial_{\alpha}\partial_{\beta}\eta=&C(g,\eta,\psi,\phi_{a}, \partial g, \partial\eta,\partial \psi),\\
g^{\alpha\beta}\partial_{\alpha}\partial_{\beta}\psi=&D(g,\eta,\psi,\phi_{a},\partial g, \partial\eta,\partial \psi).
\notag
\end{align}
Here $B_{\mu\nu}$ was defined above and 
\begin{align}
C=&-\Gamma^\alpha \partial_\alpha \psi-\frac{1}{\sqrt{3}M_{p}}\exp{(-\frac{2\eta}{\sqrt{3}M_{p}})}\partial_{\mu}\psi\partial^{\mu}\psi-\frac{4}{\sqrt{3}M_{p}}\exp(-\frac{4\eta}{\sqrt{3}M_{p}})(\epsilon-\mu^{(a)}\phi_{a})\\ \notag
&-\exp(-\frac{4\eta}{\sqrt{3}M_{p}})(\frac{\partial\epsilon}{\partial\eta}-\mu^{(n)}\frac{\partial\phi_{n}}{\partial\eta}),\\
D=&-\Gamma^\alpha \partial_\alpha \psi+\frac{2}{\sqrt{3}M_{p}}\partial_{\mu}\eta\partial^{\mu}\psi+\exp(-\frac{2\eta}{\sqrt{3}M_{p}})\frac{\partial\epsilon}{\partial\psi}.
\end{align}
This should be supplemented by the $n-1$ constraint equations
$$
\frac{\partial \epsilon}{\partial\phi_{a}}+\frac{\partial\epsilon}{\partial\phi_{n}}\frac{\partial\phi_{n}}{\partial\phi_{a}}=\mu^{a},\qquad a=1,..,n-1,
$$
and thus $\phi_a=\phi_a(\eta)$. By assuming that the auxiliary fields are eliminated as functions of $\eta$, one may define
the unknowns  $u_q$=($g_{\mu\nu}$, $\psi$, $\eta$) and $f_{q}$=($B_{\mu\nu}$, $C$, $D$). In these terms the last system can expressed as
$$
g^{\mu\nu}\frac{\partial u_q}{\partial x^\mu \partial x^\nu}=f_q(u_l,\partial u_l).
$$
This is of the form of an hyperbolic quasi linear second order system (\ref{choquet}), which is what we were looking for.

The systems of the form (\ref{choquet}) were studied in \cite{leray} and in the context of pure GR in \cite{foures}. From these references it is inferred that there exists some conditions
for which the evolution is well defined. First, one chooses an achronal surface $S$ with normal $n$ for which the initial data is settled. By use of the synchronous coordinate system 
$$
g=-dt^2+g_{ij}(t, x^i) dx^i dx^j,
$$
the initial surface can be defined by the simple equation $t=0$. The initial values will be denoted as $\phi_i=W_i$ and $\chi_i=\partial_t W_i$. By definition it is seen that both $\phi_i$ and $\psi_i$ are quantities defined on $S$. The first assumption is  that, in a domain $D$, in the initial surface $S$  defined by $|x^i-x^i_0|\leq d$ with $d$ a constant,
the functions $\phi_i$ and $\chi_i$ are differentiable up to fifth and fourth order. These functions are supposed to satisfy the Lipshitz condition $|\phi_i(x)-\phi_i(x')|\leq M |x-x'|$ and $|\chi_i(x)-\chi_i(x')|\leq M |x-x'|$ in all their arguments. Furthermore, for the values defined by 
$$
|W_i-\phi_i|\leq l, \qquad |\partial_t W_i-\chi_i|\leq l,\qquad \bigg|\frac{\partial W_i}{\partial x^j}-\frac{\partial \phi_i}{\partial x^j}\bigg|\leq l,
$$ 
in the domain 
$$
|x^i-x^i_0|\leq d,\qquad |t|\leq \epsilon,
$$
it is assumed that $g_{00}<0$ and $g_{ij}\xi^i\xi^j>0$ and that both $g_{\mu\nu}$ and $f_i$ have derivatives up to fourth orders, continuous and bounded, and satisfying the Lipschitz condition. Under these assumptions there exists an unique solution of the system of differential equations, with continuous and bounded derivatives up to fourth order, in a region 
$$
|x^i-x^i_0|\leq d,\qquad |t|\leq \eta(x').
$$ 
Note that this result does not insure that the degree of regularity of the initial condition is preserved for the evolution, since 
it ensures that the fields of the model are $C^4$ while the initial condition is assumed to be $C^5$. 

It should be emphasized that the Lipschitz condition is equivalent to a restriction of the coupling $\epsilon(\psi, \eta)$. But since we are going to prove the existence of smooth solutions below, we postpone the analysis of these restrictions till later on.
\subsection{The use of harmonic coordinates}
There is still a further aspect to be analysed. In deriving the existence result of the previous paragraph, the harmonic gauge $F^\mu=0$ has been employed at $t=0$ and it was assumed that it holds during the evolution of the resulting space time. When this gauge is taken into account, the resulting coupled Einstein system becomes a quasi linear hyperbolic one, and a local solution exits which is differentiable up to fourth order. However, this does not imply that the gauge $F^\mu=0$ will be satisfied during the evolution at $t>0$, even though if this condition holds at $t=0$. If this inconsistency appears, then the solution of (\ref{hyper system}) is not a solution of the q-theory. Such solution would be clearly unphysical. Thus, the evolution of the quantity $F^\mu$ should be analysed separately. For this, recall that the Ricci tensor $R^{\mu\nu}$ corresponding to $g_{\mu\nu}$ is explicitly
\begin{equation}\label{desco2}
R^{\mu\nu}=-\frac{1}{2}g^{\alpha\beta}\partial_{\alpha}\partial_{\beta}g^{\mu\nu}+Q^{\mu\nu}(g,\partial g)+\frac{1}{2}(g^{\beta\mu}\partial_\beta F^\nu+g^{\beta\nu}\partial_\beta F^\mu).
\end{equation} 
It is just when the harmonic gauge $F^\mu=0$ is imposed that the last expression reduces to (\ref{desco}), which is the expression employed to derive the system (\ref{hyper system}). Define
$$
R_F^{\mu\nu}=-\frac{1}{2}g^{\alpha\beta}\partial_{\alpha}\partial_{\beta}g^{\mu\nu}+Q^{\mu\nu}(g,\partial g).
$$
Then $R_F^{\mu\nu}=R^{\mu\nu}$ if the harmonic gauge $F^\mu=0$ holds. In general, the Einstein equations are equivalent to
$$
R_{F}^{\mu\nu}-\frac{1}{2}g^{\mu\nu}R_F-T^{\mu\nu}=\frac{1}{2}(g^{\mu\alpha}\partial_\alpha F^\nu+g^{\nu\alpha}\partial_\alpha F^\mu-g^{\mu\nu}\partial_\alpha F^\alpha),
$$
while the system (\ref{hyper system}) solved above implies that
\begin{equation}\label{uche}
R_{F}^{\mu\nu}-\frac{1}{2}g^{\mu\nu}R_F-T^{\mu\nu}=0.
\end{equation}
Thus, there is an inconsistency unless $F^\mu=0$ for all times where the solution exists. One tool for proving that $F^\mu=0$ for $t>0$ is the identity
$$
\nabla^\mu (R_{\mu\nu}-\frac{1}{2}g_{\mu\nu}R-T_{\mu\nu})=0,
$$
which is satisfied for a true Einstein solution. This, together (\ref{uche}) with the last identity implies that $F^\mu$ satisfies the following equation
$$
g^{\mu\nu}\partial_{\mu}\partial_{\nu}F^\alpha+A^{\alpha\beta}_\gamma \partial_\beta F^\gamma=0,
$$
for the true solution. Here $A^{\alpha\beta}_\gamma$ are some quantities depending on the metric tensor and its derivatives. The important point is that this system is also of the form
(\ref{choquet}). Now, the hamiltonian constraint $R^{0\mu}-\frac{1}{2}g^{0\mu}R-T^{0\mu}=0$ on the initial surface implies that $\partial_t F^\mu=0$ at $t=0$. This, together with  the uniqueness property for the systems (\ref{choquet}) and the  condition $F^\mu=0$ at $t=0$ imply that $F^\mu=0$ at the times $t\geq0$ where the solution of (\ref{hyper system}) exists. Thus, no inconsistency appears and the solution characterized in this section is a physical one. This is a very important point about the harmonic gauge.

\section{Local theorems on existence}

There exist a technique for reducing second order systems to first order one, which may be easier to deal with  \cite{Taylor}-\cite{courant}.
These techniques were applied for studying asymptoticall flat solutions in GR in \cite{marsden} for single scalar fields and other situations in \cite{friedrich}-\cite{friedrich2}. In addition, the Cauchy problem for Hordeski theories was analysed in \cite{reall} when the gravitational field is weak. The approach of \cite{friedrich}-\cite{friedrich2} requires to use an vierbein formalism  for the metric, and it is likely that these methods may be applied here. 
But we will not pursue these methods here. Instead we will consider the techniques for second order systems given in \cite{ringstrom}, which deals directly with the second order system.

The general form of the system (\ref{hyper system})  is the following
\begin{equation}\label{chic}
g^{\mu\nu}(x,t,u)\partial_\mu\partial_\nu u(x,t)=f(x,t,u, \partial_\alpha u).
\end{equation}
Here $u$=($g_{\mu\nu}$, $\psi$, $\eta$) represent the unknowns and $f$=($B_{\mu\nu}$, $C$, $D$) represents the non linearity. The quantities $g^{\mu\nu}(x,t, u)$
are the inverse of the metric tensor $g_{\mu\nu}$, thus, $g^{\mu\nu}( g_{\alpha\beta})$.
Denote the initial conditions by
\begin{equation}\label{chic 2}
u(x,T_0)=U_0,\qquad \partial_t u(x, T_0)=U_1.
\end{equation}
Under certain specific circumstances, it can be proved that the system (\ref{chic}) has a local solution. It is convenient at this stage to introduce the vector
$$
\xi=(g_{\mu\nu}, \partial_\alpha g_{\mu\nu},\eta, \psi,\partial_\mu \psi, \partial_\mu \eta).
$$
The introduction of this vector function facilitates the introduction of some relevant definitions and the statement of the proposition given below.

First consider the map $g_{\mu\nu}\in C^{\infty}(\R^{nN+2N+n+1},L_n)$, where $L_n$ denotes the space of canonical
$(n+1)\times (n+1)$ Lorentz matrices. Assume that these quantities satisfy
$$
|\partial^\alpha g_{\mu\nu}(x,t, \xi)|\leq h_{I, \alpha}(|\xi|),
$$
where ($x$ ,$t$) are local coordinates on $\R^{n+1}$ and $\xi$ parametrize the coordinates of $\R^{nN+2N}$. Here $I=[T_1, T_2]$ is any compact time interval and $h_{I, \alpha}: \R\to \R$ are continuous increasing functions for every  multi index $\alpha=$($\alpha_1$,..,$\alpha_{nN+2N+n+1})$. Suppose that for any compact interval $I$ there are constants $a_i\geq 0$ with $i=1,2,3$ such that
$$
g_{00}\leq -a_1,\qquad \det g_{ij}\geq a_2, \qquad \sum^n_{(\mu,\nu)=0}|g_{\mu\nu}|\leq a_3.
$$
The quantities satisfying the last condition are known as $C_{n, a}$ metrics, and the metrics satisfying all of the aformentioned assumptions are known as $C^\infty$ $N$, $n$ admissible metrics. Furthermore, for the  {\it non-linearity} $f$ is assumed that   
\begin{equation}\label{chicas}
|\partial^\alpha f(x,t, \xi)|\leq \bar{h}_{I, \alpha}(|\xi|),
\end{equation}
with $\bar{h}_{I, \alpha}(|\xi|)$ functions of the same type as the  $h_{I, \alpha}(|\xi|)$ above, and the time interval $I$ is also compact. In addition $f(x, t, \xi)$ is such that for each compact interval $I$,
there exist a compact set $K\subset \mathbb \R^3$ such that $f(x, t, 0)=0$ for any $x$ outside $K$ and $t \in I$. Such functions are known as locally of $x$-compact support. In these terms the following proposition may be stated, as in chapter 9 of the reference \cite{ringstrom}.

\begin{prop} Under conditions stated above, let $U_0, U_1 \in C^{\infty}(\R^n, \R^N)$  and $T_0\in\R$.  Then there exist two times $T_1$ and $T_2$ such that $T_1< T_0<T_2$ for which there is a unique $C^\infty$ solution $u$ of the system \eqref{chic} and \eqref{chic 2}. This solution is of  $x$-compact support.
\end{prop}

It should be emphasized that the $x$-compact support is a rather technical one. Its importance is due to the fact that a function $u:R^{n+1}\to R^m$ can be viewed as an element of $C^l[R, H^k(n, m)]$ for every value of $l$ and $k$. This plays an important role in the proof of the proposition, as it can be seen by reading the chapter 8 and 9 of \cite{ringstrom}.

Given this result, the first task is to check if the system (\ref{hyper system}) describing the q-model is of the form of Proposition 1. First of all, the quantities $B_{\mu\nu}$, $C$ and $D$ are playing the role of the quantity $f(x, t, \xi)$ of the proposition. Thus one should check if these are of $x$-compact support, which is one of the assumptions. In order to check the x-compactness consider for instance the quantity $C$ in (\ref{hyper system}). It is given explicitly by
$$
C(\xi)=C(g_{\mu\nu}, \partial_\alpha g_{\mu\nu},\eta, \psi,\partial_\mu \psi, \partial_\mu \eta)=-\Gamma^\alpha \partial_\alpha \psi+\frac{2}{\sqrt{3}M_{p}}\partial_{\mu}\eta\partial^{\mu}\psi+\exp(-\frac{2\eta}{\sqrt{3}M_{p}})\frac{\partial\epsilon}{\partial\psi}.
$$ 
The quantity $C$ does not depend explicitly on space time coordinates $(x,t)$, the dependence is implicit due to the vector function $\xi$. On the other hand, if $\xi=(0,..,0)$, it is clear that 
$$
C(0,..,0)=\frac{\partial\epsilon}{\partial\psi}\bigg|_{\psi,\eta=0}\neq 0.
$$
The same happens with the quantity $D$, which is explicitly 
$$
D=-\Gamma^\alpha \partial_\alpha \psi-\frac{1}{\sqrt{3}M_{p}}\exp{(-\frac{2\eta}{\sqrt{3}M_{p}})}\partial_{\mu}\psi\partial^{\mu}\psi-\frac{4}{\sqrt{3}M_{p}}\exp(-\frac{4\eta}{\sqrt{3}M_{p}})(\epsilon-\mu^{(a)}\phi_{a})
$$
$$
-\exp(-\frac{4\eta}{\sqrt{3}M_{p}})(\frac{\partial\epsilon}{\partial\eta}-\mu^{(n)}\frac{\partial\phi_{n}}{\partial\eta}).
$$
From here it is directly deduced that 
$$
D(0,..,0)=-\bigg[\frac{4}{\sqrt{3}M_{p}}(\epsilon-\mu^{(a)}\phi_{a})
+(\frac{\partial\epsilon}{\partial\eta}-\mu^{(n)}\frac{\partial\phi_{n}}{\partial\eta})\bigg]\bigg|_{\psi,\eta=0}\neq 0.
$$
Therefore, it seems that the $x$-support condition is spoiled in our system. However, this reasoning should be taken with care, since one can redefine the fields by making a shift $\psi\to\psi-\psi_0$, and $\eta\to\eta-\eta_0$. In particular, there exist a choice $\psi_0$, $\eta_0$ such that the quantities given above evaluated at these field values are zero. These fields are the minima of the effective potential $U_{eff}$ described in (\ref{efectivo}). If one makes the redefinition $\psi\to\psi-\psi_0$, $\eta\to\eta-\eta_0$ with the fields $\psi_0$, $\eta_0$ global minima of  (\ref{efectivo}), then the condition
$$
C(0,..,0)= D(0,..,0)=0,
$$
is satisfied. On the other and, the terms in the definition of the quantities $B_{\mu\nu}$ in (\ref{be}) are zero. These quantities are explicitly
$$
B_{\mu\nu}=G_{N}\left[\partial_{\mu}\eta\partial_{\nu}\eta+\exp(-\frac{2\eta}{\sqrt{3}M_{p}})\partial_{\mu}\psi\partial_{\nu}\psi+\exp(-\frac{4\eta}{\sqrt{3}M_{p}})g_{\mu\nu}(\epsilon-\mu^{a}\phi_{a})\right]-Q_{\mu\nu}(g,\partial g).
$$
Note that the potentially non zero term is just $(\epsilon-\mu^{a}\phi_{a})$, and it is multiplied by $g_{\mu\nu}$. The condition $\xi=(0,..,0)$ obviously implies that $g_{\mu\nu}=0$. Thus, the defining quantities $C$, $D$ and $B_{\mu\nu}$ are all of $x$-compact support when the field redefinition is performed. 

A further task is to check that (\ref{chicas}) also hold. This requires to take derivatives of any order of $C$, $B_{\mu\nu}$ and $D$. This is cumbersome but straightforward. The process of taking derivatives will result in derivatives the function $\epsilon(\psi, \eta)$ and the fields $\phi_a(\eta)$. Thus, the constraint will be satisfied if 
the function $\epsilon(\psi, \eta)$ and the fields $\phi_a(\eta)$ are reasonable, for instance, when they are continuous, smooth and do not posses vertical asymptotes at finite values of the fields $\psi$ and $\eta$. Thus, under more or less generic circumstances, this constraint will take place in the q-models.

There is however one condition that is surely violated. The aim of solving  equations (\ref{chic}) is to construct a Lorentz space-time metric $g_{\mu\nu}(x,t)$, which is never of $x$-compact support. This is not in agreement with the propositions assumptions. However, the previous proposition allows the construction of the metric $g_{\mu\nu}$ locally, and a global metric may be obtained by  a suitable gluing process. This procedure is sketched in the appendix, but full details may be found in \cite{ringstrom}.

\section{Discussion}
In the present work the Cauchy formulation of q-theories was studied in detail. It was shown the existence of a preferred frame, the Einstein frame, where the Cauchy problem is well formulated. The resulting system of equations, in the harmonic gauge, is 
a quasi-linear hyperbolic one. By use of some modern theorems about non linear second order systems we were able to prove the existence of global smooth solutions, by assuming that the space time manifold is globally hyperbolic. It should be emphasized that there exists a limit for which these theories are described by an $F(R)$ theory, for which the Cauchy problem is well defined. The contribution of this paper is that this problem is well posed even without taking this limit. Further aspects such as the existence of an maximal extension for a given solution, the Cauchy problem for non hyperbolic space times or the appearance of singularities will be considered in a separate publication.

\section*{Acknowledgments}
Both authors are supported by CONICET, Argentina. O.P.S is supported by the Beca Externa Jovenes Investigadores of CONICET. O.P.S warmly acknowledge the Steklov Mathematical Institute of the Russian Academy of Sciences in Moscow, were part of this work has been done, for their hospitality.

\appendix
\section{The existence of smooth solutions}
In the present work the equation (\ref{chic}) played a fundamental role. The left hand side of (\ref{chic}) involves the inverse metric $g^{\mu\nu}$, which is a function of the metric $g_{\mu\nu}$ to be solved. The problem is that the obtained solution $g_{\mu\nu}$, $\psi$ and $\eta$ are of x-compact support. A true space metric is not expected to be of this form. In the book \cite{ringstrom} a procedure in order to avoid this problem was presented. The idea is to multiply the equations and the initial conditions by certain suitable functions, which makes the system of compact x-support. After this, by a suitable gluing process, an smooth space time metric which is not of compact support may be constructed. Here we give some more or less detailed account of this procedure, applied to the present case, but further details may be found in the chapters 9 and 14 of the original reference \cite{ringstrom}.

\subsection{The local form of the metric}
The first task is to define an achronal surface where the initial conditions will be settled.
Recall that a given an space time  $(M, g_{\mu\nu})$ there exists the so called synchronous reference system for which the metric may be expressed as follows
$$
g=-dt^2+g_{ij}(t,x) dx^i dx^j.
$$
The coordinate $t$ represents the proper time and the spatial metric $g_{ij}(x_i,t)$  depends on $t$ as a parameter.
This system of coordinates exists locally, that is, it exists in a subset of the space time of the form $U\times I\subseteq M$ with $I=[0, t_1)$ and $t_1>0$. Consider an achronal initial Cauchy surface $\Sigma$, in such a way that the metric is regular on this surface. Given a subset $U\subseteq \Sigma$ there exist a neighbourhood $O$ with synchronous coordinates $(x_i,t)$ such that the subset $U$ is given by the surface $t=0$. Then, one may adapt the arguments given in chapter 9 of \cite{ringstrom} to the present situation as follows.

 Consider a subset $V\subseteq U$ such that its closure $\overline{V}\subseteq U$. Take a function $\widetilde{g}_{00}$ such that its range is bounded in $[-2,-1/4]$ and such that $\widetilde{g}_{00}=g_{00}$ when the value of $g_{00}$ is in the range  $[-3/2,-1/2]$. In addition, define $\widetilde{g}_{0i}$ such that its range is bounded in $[-2, 2]$ and such that $\widetilde{g}_{0i}=g_{0i}$ when the value of $g_{0i}$ is in the range  $[-1,1]$. There is nothing special about this choice of intervals, and a continuum of other choices are possible. The important point is however that the interval on which $\widetilde{g}_{0i}=g_{0i}$ should contain $0$. Moreover, the range of $\widetilde{g}_{0i}$ should contain the interval on which $\widetilde{g}_{0i}=g_{0i}$, with a margin. Finally consider an open set $R$ of symmetric $3\times 3$  matrices such that the values of $g_{ij}(x)$ with $x\in \overline{V}$ are in $R$ and the closure of $R$ in the set of $3\times 3$ matrices is compact and included in the set of positive defined ones. 
Then, one defines $\widetilde{g}_{ij}=g_{ij}$ when the value of $g_{ij}$ is in $U$ and it is assumed that $\widetilde{g}_{ij}$ has a positive lower bound and a positive upper one. Furthermore, it is assumed that the derivatives of $\widetilde{g}_{\mu\nu}$ with respect to the metric is are of compact support. Then replace $g^{\mu\nu}$ by $\widetilde{g}^{\mu\nu}$ in the system (\ref{chic}).  Also replace
the quantity $f^\mu$ by $c f^\mu$ with $c\in C_0^{\infty}[(-1,1)\times U]$ such that $c$ takes values $[-1/2, 1/2]\times \overline{V}$. The system becomes
The general form of the system (\ref{hyper system})  is the following
\begin{equation}\label{chics}
\widetilde{g}^{\mu\nu}(x,t,u)\partial_\mu\partial_\nu u(x,t)=c f^\mu(x,t,u, \partial_\alpha u).
\end{equation}
In order to apply the Proposition 1 of the text, one must modify the initial conditions by multiplying them by a function $h(x_i)$ which is $C^\infty$ and of compact support, and such that $h(q)=1$ for $q\in V$. 
\begin{equation}\label{chics 2}
u(x,T_0)=h U_0,\qquad \partial_t u(x, T_0)=h U_1.
\end{equation}
In these terms the Proposition 1 applies and one obtains a local solution. This implies that, for a given point $p$ in $\Sigma$, there exists an open neighbourhood $O$ such that there exist a solution for which $g_{\mu\nu}=\widetilde{g}_{\mu\nu}$ and $h=1$. This is the local form of the metric we were looking for, and it is of $x$-compact support.

 \subsection{Gluing the local metric to a global one}

\subsubsection{Preliminary lemmas}

The main task now is to glue the local solution described in the previous paragraph to a global one. But before doing that, it is necessary to state some results about sequence of points. Note that situation described above corresponds to a globally hyperbolic space time which admits an smooth Cauchy hypersurface $\Sigma$ and there exist a metric defined in an open set $U\subseteq M$. 

Recall that the $J^{-}(p)$ is the causal past of the point $p$, which is composed for all the points $x$ that they causally precede $p$, that is
$$
J^{-}(p)=\{x|\;\;\; x<<p\}.
$$
Analogous definition holds for $J^{+}(p)$. The chronological past and future of $p$, namely $I^{\pm}(p)$, is defined by changing the word causally by chronologically in the previous definition.  For a subset $S$ of $M$ one defines
$$
J^{\pm}(S)=\cup_{x\in S} J^{\pm}(x),
$$
and the analogous hold for $I^{\pm}(S)$. The future Cauchy development of $S$, $D^{+}(S)$  is the set of all points $x$ for which every past directed inextendible causal curve through $x$ intersects $S$ at least once. Similarly for the past Cauchy development. The Cauchy development is the union of the future and past Cauchy developments. An space is globally hyperbolic if there exists a surface $\Sigma$ such that $D(\Sigma)=D^+(\Sigma)\cup D^{-}(\Sigma)=M$. The surface $\Sigma$ is known as a Cauchy surface of $M$ and if there is one, there is a continuum of them.
In these terms, the following two lemmas apply.

\begin{lemma}\label{lema1} Given a Cauchy surface $\Sigma$ in a globally hyperbolic space time $(M, g)$, denote its Cauchy development by $D(\Sigma)=D^+(\Sigma)\cup D^{-}(\Sigma)=M$. For any  point $p\in$ Int $D(\Sigma)-I^{-}(\Sigma)$ the set $J^{-}(p)\cap J^{+}(\Sigma)$ is compact. 
\end{lemma}

The statement of Lemma \ref{lema1} is intuitive by analyzing it in the case of a Minkowski space time, as in this case, the resulting set $J^{-}(p)\cap J^{+}(\Sigma)$ is the intersection of two compact spaces. However, the proof of this statement is not that straightforward for generic globally hyperbolic space times, as it requires to understand the infinite dimensional space of causal curves $C(\Sigma, p)$ connecting $\Sigma$ with $p$ with a given second countable Haussdorf topology and, in particular, to show that it is compact. Details are given in  the books \cite{Wald}-\cite{O Neill}.

\begin{lemma}\label{lema2} For a given globally hyperbolic space time $(M, g)$ with a smooth Cauchy hypersurface $\Sigma$, consider an open set $U\in M$ and a point $q$ such that $J^+(S)\cap J^-(q)\in U$. Given a sequence $q_i\to q$ then $J^+(S)\cap J^-(q_i)\subset U$ for $i\geq i_0$.
\end{lemma}

\begin{proof} Consider a point $q$ which belongs to the causal future $J^{+}(\Sigma)$ of $\Sigma$  and such that $J^+(\Sigma)\cap J^-(q)\subseteq U$.  Then, if there is a future directed time like curve $\gamma$ which connects $q$ with another generic point, say $p$, it follows that $q$ is in the interior of $J^{-}(p)$. Furthermore lemma 1 shows that  $J^{-}(p)\cap J^+(\Sigma)$ is a compact set, and this will be exploited to prove the assertion. Consider a sequence of points $q_i$ in $J^+(\Sigma)$ such that $q_i\to q$, then it is not difficult to see that $J^{-}(q_i)\subseteq J^{-}(p)$ when $i>i_0$. This follows from the fact that $q$ is the accumulation point of the sequence, and for $i$ large enough, these points will be in the causal past of $p$, as both $p$ and $q$ are connected by a \emph{time like} curve.  It is also intuitive that  $J^{-}(q_i)\cap J^{+}(\Sigma)\subseteq U$ when $i>i_1$, since we are assuming that $J^+(\Sigma)\cap J^-(q)\subseteq U$. In fact, suppose that there were a subsequence of points $q_l$ such that the corresponding set $J^{-}(q_l)\cap J^{+}(\Sigma)$ contains point $r_l$ which are outside $U$ even for $l$ large. Then these points $r_l$ are located in  $J^{-}(p)\cap J^{+}(\Sigma)-U$, which is a compact set as it is a compact space with a deleted space open space $U$. Thus $J^{-}(p)\cap J^{+}(\Sigma)-U$ contains its accumulation points and the sequence $r_l$ converges to a point $r$. Every point $r_l\in J^{-}(q_l)\cap J^{+}(\Sigma)-U$ and therefore $r\in J^{-}(q)\cap J^{+}(\Sigma)-U$. But $J^+(\Sigma)\cap J^-(q)\subseteq U$ and thus, the last statement is inconsistent. This contradiction shows that, in fact, $J^+(S)\cap J^-(q_i)\subset U$ for $i\geq i_0$ and this is precisely the assertion we wanted to prove.
\end{proof}

\subsubsection{The gluing process}

Equipped with this lemma, let us return now to the gluing process. The following is an adaptation of some arguments of  \cite{ringstrom} to the present situation. For this, let $W_p$ an open neighbour of $p$ such that its closure is $W$. Then consider the manifold $M=\cup_{p}W_p$. Given two sets $W_p$ and $W_q$ there are two fields $u_1$ and $u_2$ which are solutions of the corresponding equations. The harmonic coordinate equation is satisfied in both systems. The initial data also coincide in both systems. The main task is to show that the solutions coincide in $\overline{W}_p\cap \overline{W}_q$. For this one let us define the time interval $I\in [0,\infty)$ for which both solutions coincide in 
$$
S_t=[0,t]\times \Sigma\cap \overline{W}_p\cap \overline{W}_q,
$$
and also such that
\begin{equation}\label{intersecto}
J_p^{-}(x)\cap J_p^{+}(\Sigma)=J_q^{-}(x)\cap J_q^{+}(\Sigma),
\end{equation}
for any $x\in S_t$. The strategy is to prove that $I$ is open and closed, and non empty, thus it is the full interval where both solutions are defined. To prove that is not empty is immediate. The initial conditions coincide, then $I$ contains the point $t=0$, which is enough to show that it is not empty.
On the other hand, the set (\ref{intersecto}) is compact due to the lemma 1, and a bit of reasoning implies that $I$ should be closed as well.
Next, one may prove that $I$ is open. For this, let $x=(t,r)$ be such that $J_p^{-}(x)\cap J_p^{+}(\Sigma)\subseteq W_p\cap W_q$. Then the  lemma 2 can be applied with the open subset $W_p\cap W_q$ playing the role of $U$. From this lemma, it is concluded that for the point $x+\delta x=(t+\epsilon, r +\delta(\epsilon))$ one has that $J_p^{-}(x+\delta x)\cap J_p^{+}(\Sigma)\subseteq W_p\cap W_q$ if $\epsilon$ is small enough. Thus, the space time extends to the point $x+\delta x$. However, this conclusion does not warrant that the two fields $u_1$ and $u_2$ coincide for $x+\delta x$. To prove that this is indeed the case, take the point $t$ as a new initial condition. This is valid, since the solutions $u_1$ and $u_2$ are known to coincide by our assumption up to $t$. Then take the difference between the two solutions $u_1$ and $u_2$ at $t$. This difference $u_1-u_2$ is zero a $t$ and their derivatives are also zero at $t$. Thus, by general results about quasi linear hyperbolic systems, it follows that $u_1-u_2=0$ up to the time where this solutions exist. Thus $u_1=u_2$ up to 
$$
J_p^{-}(t+\epsilon, r+\delta(\epsilon))\cap J_p^{+}(\Sigma)=J_q^{-}(t+\epsilon, r+\delta(\epsilon))\cap J_q^{+}(\Sigma).
$$
This implies that, given a time $t\in I$, then $[t, t+\epsilon]\subseteq I$ for $\epsilon$ small enough. This means that $I$ is an open set. 

From the paragraph given above, it is clear that $I$ is empty, open and closed, thus $I=[0,\infty)$. This shows that the constructed metric glues properly on $M=\cup_{p}W_p$, up to a point where a singularity appears, or at all times if the universe is future eternal.

\end{document}